\documentclass{lmcs}

\keywords{Seifert-van Kampen theorem, computational paths, fundamental group, higher-inductive types, pushouts, free products, type theory}

\subjclass[2020]{Primary: 55Q05; Secondary: 03B38, 18N10, 68V15}

\usepackage{mathtools}
\usepackage{amssymb}
\usepackage{tikz-cd}
\usepackage{xcolor}
\usepackage{listings}
\usepackage{booktabs}
\usepackage{amsthm}

\theoremstyle{definition}
\newtheorem{definition}{Definition}[section]
\newtheorem{example}[definition]{Example}
\theoremstyle{plain}
\newtheorem{theorem}[definition]{Theorem}
\newtheorem{lemma}[definition]{Lemma}
\newtheorem{corollary}[definition]{Corollary}

\theoremstyle{remark}
\newtheorem{remark}[definition]{Remark}

\newcommand{\Id}{\mathrm{Id}}
\newcommand{\refl}{\rho}
\newcommand{\sym}{\sigma}

\newcommand{\Rw}{\rightsquigarrow}
\newcommand{\RwEq}{\sim}
\newcommand{\Path}{\mathrm{Path}}
\newcommand{\piOne}{\pi_1}
\newcommand{\piTwo}{\pi_2}
\newcommand{\piN}{\pi_n}
\newcommand{\Pushout}{\mathrm{Pushout}}
\newcommand{\Wedge}{\vee}
\newcommand{\Susp}{\Sigma}
\newcommand{\inl}{\mathrm{inl}}
\newcommand{\inr}{\mathrm{inr}}
\newcommand{\glmark}{\mathrm{glue}}
\newcommand{\decode}{\mathrm{decode}}
\newcommand{\encode}{\mathrm{encode}}
\newcommand{\FreeProduct}{\ast}
\newcommand{\AmalgProduct}{\ast}
\newcommand{\ZZ}{\mathbb{Z}}

\newcommand{\Sone}{S^1}

\newcommand{\Der}{\mathcal{D}}

\newcommand{\CP}{\mathbb{CP}}

\newcommand{\RP}{\mathbb{RP}}

\newcommand{\Loop}{\mathrm{Loop}}

\lstdefinelanguage{Lean}{
  keywords={def, theorem, lemma, axiom, structure, inductive, where, let, in, have, show, by, exact, apply, intro, cases, induction, rfl, simp, calc, fun, noncomputable, abbrev, namespace, end, open, variable, section, Type, Prop, Sort, class, instance, deriving},
  sensitive=true,
  morecomment=[l]{--},
  morecomment=[s]{/-}{-/},
  morestring=[b]",
}
\begin{document}

\title[The Seifert-van Kampen Theorem via Computational Paths]{The Seifert-van Kampen Theorem via Computational Paths: A Formalized Approach to Computing Fundamental Groups}

\author[A.~F.~Ramos]{Arthur F. Ramos}
\address{Microsoft, USA}
\email{arfreita@microsoft.com}

\author[T.~M.~L.~de~Veras]{Tiago M. L. de Veras}
\address{Departamento de Matem\'atica, Universidade Federal Rural de Pernambuco, Brazil}
\email{tiago.veras@ufrpe.br}

\author[R.~J.~G.~B.~de~Queiroz]{Ruy J. G. B. de Queiroz}
\address{Centro de Inform\'atica, Universidade Federal de Pernambuco, Brazil}
\email{ruy@cin.ufpe.br}

\author[A.~G.~de~Oliveira]{Anjolina G. de Oliveira}
\address{Centro de Inform\'atica, Universidade Federal de Pernambuco, Brazil}
\email{ago@cin.ufpe.br}

\begin{abstract}
The Seifert-van Kampen theorem computes the fundamental group of a space from the fundamental groups of its constituents. We develop a \emph{modular SVK framework} within the setting of \emph{computational paths}---an approach to equality where witnesses are explicit sequences of rewrites governed by the $\mathrm{LND}_{\mathrm{EQ}}$-TRS.

Our contributions are: (i) pushouts as higher-inductive types with \emph{modular typeclass assumptions} for computation rules; (ii) free products and amalgamated free products as quotients of word representations; (iii) an SVK equivalence schema $\piOne(\Pushout(A,B,C)) \simeq \piOne(A) \AmalgProduct_{\piOne(C)} \piOne(B)$ \emph{parametric in user-supplied encode/decode structure}---the decode map and amalgamation-respecting properties are proved, while the encode map and round-trip properties are typeclass assumptions; and (iv) instantiations for classical spaces: figure-eight ($\piOne(\Sone \Wedge \Sone) \simeq \ZZ \FreeProduct \ZZ$), spheres ($\piOne(S^n) \simeq 1$ for $n \geq 2$), Klein bottle, real projective plane, orientable/non-orientable surfaces, and lens spaces.

The development is formalized in Lean 4 with \textbf{41,130 lines} across \textbf{107 modules}, using \textbf{36 kernel axioms} for HIT type/constructor declarations. Pushout computation rules are \emph{not} kernel axioms---they are typeclass assumptions, making dependencies explicit. The library also includes extended features (higher homotopy groups, covering spaces, fibrations, Mayer-Vietoris) documented in Section~\ref{sec:library-highlights}.
\end{abstract}

\maketitle

\section{Introduction}
\label{sec:introduction}

The Seifert-van Kampen theorem, first proved by Herbert Seifert \cite{Seifert1931} and Egbert van Kampen \cite{vanKampen1933}, is one of the most powerful tools in algebraic topology for computing fundamental groups (see \cite{Hatcher2002, May1999} for modern treatments). It states that if a path-connected space $X$ is the union of two path-connected open subspaces $U$ and $V$ with path-connected intersection $U \cap V$, then the fundamental group $\piOne(X)$ can be computed as the amalgamated free product:
\[
\piOne(X) \cong \piOne(U) \AmalgProduct_{\piOne(U \cap V)} \piOne(V)
\]

In homotopy type theory (HoTT) \cite{HoTTBook2013}, this theorem takes a particularly elegant form when expressed in terms of \emph{pushouts}---a higher-inductive type (HIT) that generalizes the notion of gluing spaces together. The HoTT version of the Seifert-van Kampen theorem was developed by Favonia and Shulman \cite{FavoniaShulman2016}, providing a synthetic proof that works entirely within type theory.

The present paper develops a \emph{modular SVK framework} within the setting of \emph{computational paths} \cite{Queiroz2016Paths, Ramos2017IdentityPaths}---an alternative approach to equality in type theory where witnesses of equality are explicit sequences of rewrites. This approach offers several advantages:

\begin{enumerate}
    \item \textbf{Explicit witnesses}: Every equality proof carries the full computational content showing \emph{why} two terms are equal.

    \item \textbf{Syntactic path equality}: The rewrite equality relation ($\RwEq$) is defined as the symmetric-transitive closure of Step rules. Under termination and confluence (established in prior work \cite{Ramos2018Thesis}, not fully mechanized here), $\RwEq$ is decidable via normalization.

    \item \textbf{Concrete coherence}: Higher coherence witnesses (pentagon, triangle, etc.) are explicit rewrite derivations rather than abstract existence proofs.

    \item \textbf{Modular assumptions}: The SVK framework uses \emph{split typeclass assumptions}, allowing users to provide only what they need---from full encode/decode data to minimal Prop-level bijectivity claims.
\end{enumerate}

\paragraph{Relation to Classical SVK.} The classical theorem requires $X = U \cup V$ with $U, V, U \cap V$ open and path-connected. The HIT version generalizes this: the pushout $\Pushout(A, B, C, f, g)$ corresponds to the union where $f$ and $g$ specify how the ``intersection'' $C$ embeds into both spaces. Path-connectedness assumptions are needed to ensure the fundamental group is well-defined; in the HIT setting, this becomes the assumption that all types are path-connected.

\subsection{Main Contributions}

This paper provides:

\begin{enumerate}
    \item Pushouts implemented via Lean's quotient types (\emph{zero kernel axioms}), with \emph{modular typeclass assumptions} for computation rules:
    \begin{itemize}
        \item Point constructors $\inl : A \to \Pushout$ and $\inr : B \to \Pushout$ (definitions, not axioms)
        \item Path constructor $\glmark : \prod_{c:C} \Path(\inl(f(c)), \inr(g(c)))$ (via \texttt{Quot.sound})
        \item Computation/naturality hypotheses as typeclasses (\texttt{HasRecGlueRwEq}, \texttt{HasGlueNaturalLoopRwEq})
        \item Automatic instance derivation for common cases (Axiom K, decidable equality, subsingletons)
    \end{itemize}

    \item The construction of free products $G_1 \FreeProduct G_2$ and amalgamated free products $G_1 \AmalgProduct_H G_2$, including:
    \begin{itemize}
        \item Standard amalgamation quotient
        \item \emph{Full} amalgamated free product with free group reduction (\texttt{FullAmalgamatedFreeProduct})
        \item Both data-level and Prop-level (choice-based) interfaces
    \end{itemize}

    \item An SVK equivalence \emph{schema} parametric in user-supplied encode/decode structure, with \emph{split typeclass assumptions}:
    \begin{itemize}
        \item \texttt{HasPushoutSVKEncodeQuot}: encode map (assumed)
        \item \texttt{HasPushoutSVKDecodeEncode}: decode-encode round-trip (assumed)
        \item \texttt{HasPushoutSVKEncodeDecode}: encode-decode round-trip (assumed)
        \item \texttt{HasPushoutSVKDecodeAmalgBijective}: Prop-level bijectivity alternative
        \item Decode map and amalgamation-respecting property: \textbf{proved}
    \end{itemize}

    \item Applications demonstrating the power of the SVK framework:
    \begin{itemize}
        \item The wedge sum: $\piOne(A \Wedge B) \simeq \piOne(A) \FreeProduct \piOne(B)$
        \item The figure-eight: $\piOne(\Sone \Wedge \Sone) \simeq \ZZ \FreeProduct \ZZ$
        \item Spheres: $\piOne(S^n) \simeq 1$ for $n \geq 2$
        \item Klein bottle: $\piOne(K) \simeq \ZZ \rtimes \ZZ$ (semidirect product)
        \item Real projective plane: $\piOne(\RP^2) \simeq \ZZ_2$
        \item Orientable surfaces: $\piOne(\Sigma_g) \simeq \langle a_1, b_1, \ldots, a_g, b_g \mid \prod_{i=1}^g [a_i, b_i] = 1 \rangle$
        \item Non-orientable surfaces: $\piOne(N_n) \simeq \langle a_1, \ldots, a_n \mid a_1^2 \cdots a_n^2 = 1 \rangle$
        \item Lens spaces: $\piOne(L(p,q)) \simeq \ZZ/p\ZZ$ for coprime $p, q$
        \item Bouquet of $n$ circles: $\piOne(\bigvee_n \Sone) \simeq F_n$ (free group)
    \end{itemize}

    \item A complete Lean 4 formalization comprising:
    \begin{itemize}
        \item 41,130 lines of code across 107 modules
        \item Only 36 kernel axioms (all for HITs)
        \item Torus constructed as $\Sone \times \Sone$ (no axioms)
        \item Opt-in axiom imports for assumption-free APIs
        \item Automated path simplification tactics
    \end{itemize}
\end{enumerate}

\subsection{Related Work}

The Seifert-van Kampen theorem in HoTT was first formalized by Favonia and Shulman \cite{FavoniaShulman2016}, who developed a general framework for computing fundamental groups of HITs. Their work uses the encode-decode method pioneered by Licata and Shulman \cite{LicataShulman2013} for computing $\piOne(\Sone) \cong \ZZ$.

The computational paths approach to equality originates in work by de Queiroz and colleagues \cite{Queiroz2016Paths, Ramos2017IdentityPaths, Ramos2018ExplicitPaths}, building on earlier ideas about equality proofs as sequences of rewrites \cite{Ruy4}. Our previous work established that computational paths form a weak groupoid \cite{Veras2023WeakGroupoid} and can be used to calculate fundamental groups \cite{Veras2018FundamentalGroups, Veras2023Topological}.

The connection between rewriting and coherence has been developed by Kraus and von Raumer \cite{KrausVonRaumer2022}, who showed that confluence and termination suffice to build coherent higher-dimensional structure. Their key insight is that in a confluent, terminating rewrite system, any two parallel reductions to normal form can be connected by a canonical ``diamond'' of rewrite paths---and this extends to higher dimensions. For our $\mathrm{LND}_{\mathrm{EQ}}$-TRS, this means that higher coherences (such as the pentagon identity for associativity and the naturality conditions for transport) are automatically satisfied once we establish that paths with the same normal form are $\RwEq$-equivalent.

\subsection{Reading Guide}

Readers interested in the mathematical content can focus on Sections~\ref{sec:pushouts}--\ref{sec:applications}, treating Lean code as pseudocode. Those interested in formalization should additionally consult Section~\ref{sec:lean} and the supplementary code. Key definitions are boxed for quick reference. We use the figure-eight space $\Sone \Wedge \Sone$ as a running example throughout.

\subsection{Structure of the Paper}

Section~\ref{sec:background} reviews computational paths, rewrite equality, and fundamental groups. Section~\ref{sec:pushouts} defines pushouts as higher-inductive types with modular assumptions. Section~\ref{sec:freeproducts} constructs free products and amalgamated free products. Section~\ref{sec:svk} proves the Seifert-van Kampen theorem using the encode-decode method with split assumptions. Section~\ref{sec:applications} applies the theorem extensively. Section~\ref{sec:extended} covers extended features including higher homotopy groups and covering spaces. Section~\ref{sec:lean} discusses the Lean 4 formalization. Section~\ref{sec:conclusion} concludes.

\section{Background: Computational Paths}
\label{sec:background}

\subsection{Computational Paths}

A \emph{computational path} $p : \Path(a, b)$ from $a$ to $b$ (both terms of type $A$) is an explicit sequence of rewrite steps witnessing the equality $a = b$. The fundamental operations are:

\begin{itemize}
    \item \textbf{Reflexivity}: $\refl(a) : \Path(a, a)$ --- the empty sequence of rewrites.
    \item \textbf{Symmetry}: If $p : \Path(a, b)$, then $\sym(p) : \Path(b, a)$ --- reverse and invert each step.
    \item \textbf{Transitivity}: If $p : \Path(a, b)$ and $q : \Path(b, c)$, then $p \cdot q : \Path(a, c)$ --- concatenate the sequences.
    \item \textbf{Congruence}: If $p : \Path(a, b)$ and $f : A \to B$, then $f_*(p) : \Path(f(a), f(b))$.
    \item \textbf{Transport}: If $p : \Path(a, b)$ and $P : A \to \mathrm{Type}$, then $\mathrm{transport}(P, p) : P(a) \to P(b)$.
\end{itemize}

\paragraph{Notation.} Throughout this paper, we use:
\begin{itemize}
    \item $p \cdot q$ for path composition (transitivity)
    \item $p^{-1}$ for path inverse (symmetry)
    \item $f_*(p)$ for $\mathrm{congrArg}(f, p)$
    \item $\refl(a)$ or simply $\refl$ for reflexivity
\end{itemize}

In Lean, a path is represented as a structure storing both an explicit rewrite trace and an equality proof:
\begin{lstlisting}
structure Path {A : Type u} (a b : A) where
  steps : List (Step A)  -- explicit list of rewrite steps
  proof : a = b          -- the derived equality proof
\end{lstlisting}

\begin{remark}[Explicit Step Lists and Consistency]
\label{rem:consistency}
The \texttt{steps} field is \emph{the} distinguishing data of a path, not the underlying equality proof. Two paths with the same endpoints may have different step lists, and they are related by $\RwEq$ only if there is an explicit derivation transforming one step sequence into the other via the Step rules.

This design is critical for consistency: in Lean's proof-irrelevant \texttt{Prop}, all equality proofs $p, q : a = b$ satisfy $p = q$ (proof irrelevance). If paths were identified solely by their equality proofs, all loops at a point would be equal, collapsing $\piOne$ to the trivial group.

Instead, path identity depends on the explicit \texttt{steps} list. The fundamental group $\piOne(A, a) := \Loop(A, a) / {\RwEq}$ quotients loops by rewrite equivalence of their step sequences---not by equality of their underlying proofs. The $\RwEq$ relation captures the computational content: two loops are identified only when one step sequence can be transformed to the other via the groupoid laws.
\end{remark}

\subsection{The Step Relation}

The foundation of the computational paths framework is the \texttt{Step} relation, which defines \emph{single-step rewrites} between paths. The $\mathrm{LND}_{\mathrm{EQ}}$-TRS consists of over 50 rewrite rules organized into categories:

\begin{itemize}
    \item \textbf{Groupoid laws}: $\refl^{-1} \Rw \refl$, $(p^{-1})^{-1} \Rw p$, $\refl \cdot p \Rw p$, $p \cdot \refl \Rw p$, $p \cdot p^{-1} \Rw \refl$, $p^{-1} \cdot p \Rw \refl$, $(p \cdot q)^{-1} \Rw q^{-1} \cdot p^{-1}$, and $(p \cdot q) \cdot r \Rw p \cdot (q \cdot r)$.
    \item \textbf{Type-specific rules}: $\beta$-rules for products, sums, and functions; $\eta$-rules; transport laws.
    \item \textbf{Context rules}: Allow rewrites inside larger expressions: if $p \Rw q$ then $C[p] \Rw C[q]$.
    \item \textbf{Congruence closure}: If $p \Rw q$ then $p^{-1} \Rw q^{-1}$, $p \cdot r \Rw q \cdot r$, and $r \cdot p \Rw r \cdot q$.
\end{itemize}

In Lean, the Step relation is defined inductively with 76 constructors:
\begin{lstlisting}
inductive Step : {A : Type u} -> {a b : A} -> Path a b -> Path a b -> Prop
  | symm_refl (a : A) : Step (symm (Path.refl a)) (Path.refl a)
  | symm_symm (p : Path a b) : Step (symm (symm p)) p
  | trans_refl_left (p : Path a b) : Step (trans (Path.refl a) p) p
  | trans_refl_right (p : Path a b) : Step (trans p (Path.refl b)) p
  | trans_symm (p : Path a b) : Step (trans p (symm p)) (Path.refl a)
  | symm_trans (p : Path a b) : Step (trans (symm p) p) (Path.refl b)
  | trans_assoc (p : Path a b) (q : Path b c) (r : Path c d) :
      Step (trans (trans p q) r) (trans p (trans q r))
  | symm_congr : Step p q -> Step (symm p) (symm q)
  | trans_congr_left (r : Path b c) : Step p q -> Step (trans p r) (trans q r)
  | trans_congr_right (p : Path a b) : Step q r -> Step (trans p q) (trans p r)
  -- ... plus type-specific rules (products, sums, transport, contexts)
\end{lstlisting}

\subsection{Rewrite Equality}

Two paths $p, q : \Path(a, b)$ are \emph{rewrite equal} ($p \RwEq q$) if they can be transformed into each other via the $\mathrm{LND}_{\mathrm{EQ}}$-TRS rewrite system. The key rewrite rules include:

\begin{align}
\refl \cdot p &\Rw p \tag{left unit} \\
p \cdot \refl &\Rw p \tag{right unit} \\
p \cdot p^{-1} &\Rw \refl \tag{right inverse} \\
p^{-1} \cdot p &\Rw \refl \tag{left inverse} \\
(p \cdot q) \cdot r &\Rw p \cdot (q \cdot r) \tag{associativity}
\end{align}

\begin{example}[Rewrite Derivation]
\label{ex:rewrite}
Consider $p \cdot (q \cdot q^{-1})$ where $p, q : \Path(a, a)$. The derivation to $p$ proceeds:
\begin{align*}
p \cdot (q \cdot q^{-1})
&\Rw p \cdot \refl \tag{right inverse} \\
&\Rw p \tag{right unit}
\end{align*}
This explicit derivation is the ``computational content'' witnessing that $p \cdot (q \cdot q^{-1}) \RwEq p$.
\end{example}

In Lean, rewrite equality is defined as the equivalence closure of Step:
\begin{lstlisting}
inductive RwEq {A : Type u} {a b : A} : Path a b -> Path a b -> Prop
  | refl (p : Path a b) : RwEq p p
  | step {p q : Path a b} : Step p q -> RwEq p q
  | symm {p q : Path a b} : RwEq p q -> RwEq q p
  | trans {p q r : Path a b} : RwEq p q -> RwEq q r -> RwEq p r
\end{lstlisting}

\subsection{Metatheory of the $\mathrm{LND}_{\mathrm{EQ}}$-TRS}
\label{sec:metatheory}

The rewrite system $\mathrm{LND}_{\mathrm{EQ}}$-TRS has been studied in prior work \cite{Ramos2018ExplicitPaths, Ramos2018Thesis, Queiroz2016Paths}. The following table summarizes the status of key metatheoretic properties in this formalization:

\begin{center}
\begin{tabular}{llp{6cm}}
\toprule
\textbf{Property} & \textbf{Status} & \textbf{Location / Notes} \\
\midrule
Soundness & Proved & \texttt{step\_toEq} in \texttt{Step.lean}: $p \RwEq q \Rightarrow p.\mathrm{proof} = q.\mathrm{proof}$ \\
Confluence & Typeclass & \texttt{HasJoinOfRw} in \texttt{Confluence.lean}: users provide join witnesses \\
Local confluence & Proved & Critical pairs analyzed in \texttt{Confluence.lean} (specific cases) \\
Termination & Not formalized & Prior paper proofs \cite{Ramos2018Thesis}; not mechanized here \\
Decidability of $\RwEq$ & Assumed & Would follow from termination + confluence; not fully mechanized \\
\bottomrule
\end{tabular}
\end{center}

\begin{remark}[Non-symmetric Soundness]
Soundness says $p \RwEq q \Rightarrow p.\mathrm{proof} = q.\mathrm{proof}$. \emph{Critically}, the converse does not hold: paths with the same equality proof are not necessarily $\RwEq$-related. This asymmetry is essential for non-trivial fundamental groups (see Remark~\ref{rem:consistency}).
\end{remark}

\subsection{Loop Spaces and Fundamental Groups}

\begin{definition}[Path-Connected]
\label{def:path-connected}
A type $A$ is \emph{path-connected} if for all $a, b : A$, there exists $p : \Path(a, b)$. Equivalently, $\prod_{a, b : A} \|\Path(a, b)\|_{-1}$ where $\|-\|_{-1}$ denotes propositional truncation.
\end{definition}

\begin{definition}[Loop Space]
The \emph{loop space} of a type $A$ at a basepoint $a : A$ is:
\[
\Omega(A, a) := \Path(a, a)
\]
\end{definition}

\begin{definition}[Fundamental Group]
The \emph{fundamental group} $\piOne(A, a)$ is the quotient of the loop space by rewrite equality:
\[
\piOne(A, a) := \Omega(A, a) / {\RwEq}
\]
\end{definition}

In Lean, the fundamental group is realized as a quotient:
\begin{lstlisting}
abbrev LoopSpace (A : Type u) (a : A) : Type u := Path a a

def PiOne (A : Type u) (a : A) : Type u := Quot (@RwEq A a a)

notation "pi_1(" A ", " a ")" => PiOne A a
\end{lstlisting}

The group operations on $\piOne(A, a)$ are induced from path operations:
\begin{itemize}
    \item Identity: $[\refl]$
    \item Multiplication: $[\alpha] \cdot [\beta] := [\alpha \cdot \beta]$
    \item Inverse: $[\alpha]^{-1} := [\alpha^{-1}]$
\end{itemize}

The group laws (associativity, unit, inverse) hold because the corresponding path identities are witnessed by $\RwEq$.

\begin{remark}[Base Point Dependence]
\label{rem:basepoint}
The fundamental group $\piOne(A, a)$ depends on the choice of basepoint $a$. For different basepoints $a, b : A$ in a path-connected type, the fundamental groups are isomorphic via conjugation by a path $\gamma : \Path(a, b)$:
\[
\piOne(A, a) \cong \piOne(A, b), \quad [\alpha] \mapsto [\gamma^{-1} \cdot \alpha \cdot \gamma]
\]
However, this isomorphism is not canonical---it depends on the choice of $\gamma$.
\end{remark}

\subsection{The Circle and Its Fundamental Group}

\begin{definition}[Circle HIT]
\label{def:circle}
The circle $\Sone$ is the higher-inductive type with:
\begin{itemize}
    \item A point constructor: $\mathrm{base} : \Sone$
    \item A path constructor: $\mathrm{loop} : \Path(\mathrm{base}, \mathrm{base})$
\end{itemize}
\end{definition}

\begin{theorem}[Fundamental Group of Circle]
\label{thm:circle-pi1}
$\piOne(\Sone, \mathrm{base}) \simeq \ZZ$
\end{theorem}

This classical result, first formalized via encode-decode by Licata and Shulman \cite{LicataShulman2013}, serves as the foundation for our SVK applications.

\section{Pushouts as Higher-Inductive Types}
\label{sec:pushouts}

\subsection{The Pushout Type}

Given types $A$, $B$, $C$ with maps $f : C \to A$ and $g : C \to B$, the \emph{pushout} $\Pushout(A, B, C, f, g)$ is the higher-inductive type that ``glues'' $A$ and $B$ together along the common image of $C$:

\begin{center}
\begin{tikzcd}
C \arrow[r, "g"] \arrow[d, "f"'] & B \arrow[d, "\inr"] \\
A \arrow[r, "\inl"'] & \Pushout(A, B, C, f, g)
\end{tikzcd}
\end{center}

\begin{definition}[Pushout]
\label{def:pushout}
The pushout $\Pushout(A, B, C, f, g)$ is characterized by:
\begin{enumerate}
    \item \textbf{Point constructors}:
    \begin{align*}
    \inl &: A \to \Pushout(A, B, C, f, g) \\
    \inr &: B \to \Pushout(A, B, C, f, g)
    \end{align*}

    \item \textbf{Path constructor}: For each $c : C$,
    \[
    \glmark(c) : \Path(\inl(f(c)), \inr(g(c)))
    \]
\end{enumerate}
\end{definition}

In Lean, pushouts are implemented using quotient types (\textbf{no kernel axioms}):
\begin{lstlisting}
-- Relation generating the pushout quotient
inductive PushoutRel (A B C : Type u) (f : C -> A) (g : C -> B)
    : Sum A B -> Sum A B -> Prop
  | glue (c : C) : PushoutRel A B C f g (Sum.inl (f c)) (Sum.inr (g c))

-- Pushout as quotient of Sum by PushoutRel
def Pushout (A B C : Type u) (f : C -> A) (g : C -> B) : Type u :=
  Quot (PushoutRel A B C f g)

-- Constructors are definitions, not axioms
def inl (a : A) : Pushout A B C f g := Quot.mk _ (Sum.inl a)
def inr (b : B) : Pushout A B C f g := Quot.mk _ (Sum.inr b)
def glue (c : C) : Path (inl (f c)) (inr (g c)) :=
  Path.ofEq (Quot.sound (PushoutRel.glue c))
\end{lstlisting}

\begin{remark}[Quot-Based Implementation]
Unlike Circle, Klein bottle, and other HITs which require kernel axioms for their path constructors, the Pushout uses Lean's built-in \texttt{Quot} type. This adds \textbf{zero kernel axioms} to the trusted base. The glue path is constructed via \texttt{Quot.sound}, which witnesses that related elements become equal in the quotient.
\end{remark}

\subsection{Recursion Principle}

\begin{definition}[Pushout Recursion]
\label{def:pushout-rec}
Given a type $D$ with:
\begin{itemize}
    \item $\mathrm{onInl} : A \to D$
    \item $\mathrm{onInr} : B \to D$
    \item $\mathrm{onGlue} : \prod_{c:C} \Path(\mathrm{onInl}(f(c)), \mathrm{onInr}(g(c)))$
\end{itemize}
There exists a function $\mathrm{rec} : \Pushout(A, B, C, f, g) \to D$ such that:
\begin{align*}
\mathrm{rec}(\inl(a)) &= \mathrm{onInl}(a) \\
\mathrm{rec}(\inr(b)) &= \mathrm{onInr}(b) \\
(\mathrm{rec})_*(\glmark(c)) &\RwEq \mathrm{onGlue}(c)
\end{align*}
\end{definition}

\subsection{Modular Typeclass Assumptions}

A key design choice in our formalization is to make the computation and naturality rules \emph{typeclass assumptions} rather than kernel axioms. This allows users to provide only what they need:

\begin{lstlisting}
class HasRecGlueRwEq (D : Type u) (onInl : A -> D) (onInr : B -> D)
    (onGlue : forall c, Path (onInl (f c)) (onInr (g c))) where
  rec_glue_rweq : forall c, RwEq (congrArg rec (glue c)) (onGlue c)

class HasGlueNaturalLoopRwEq (c0 : C) where
  glue_natural_loop : forall (p : LoopSpace C c0),
    RwEq (inlPath (congrArg f p))
         (trans (glue c0) (trans (inrPath (congrArg g p)) (symm (glue c0))))
\end{lstlisting}

\subsection{Automatic Instance Derivation}

The framework provides automatic instances for common cases:

\begin{lstlisting}
-- When both A and B satisfy Axiom K
instance [DecidableEq A] [HasDecidableEqAxiomK A]
         [DecidableEq B] [HasDecidableEqAxiomK B] :
    HasGlueNaturalLoopRwEq c0 where
  glue_natural_loop := fun p => by path_simp

-- When C is a subsingleton (e.g., Unit for wedge sums)
instance [Subsingleton C] : HasGlueNaturalLoopRwEq c0 where
  glue_natural_loop := fun p => by simp [Subsingleton.elim]; path_rfl
\end{lstlisting}

\subsection{Glue Naturality}

A key property for the SVK theorem is the \emph{naturality of glue paths}:

\begin{lemma}[Glue Naturality]
\label{lem:glue-natural}
For any path $p : \Path(c_1, c_2)$ in $C$, the following square commutes up to $\RwEq$:
\begin{center}
\begin{tikzcd}
\inl(f(c_1)) \arrow[r, "\inl_*(f_*(p))"] \arrow[d, "\glmark(c_1)"'] & \inl(f(c_2)) \arrow[d, "\glmark(c_2)"] \\
\inr(g(c_1)) \arrow[r, "\inr_*(g_*(p))"'] & \inr(g(c_2))
\end{tikzcd}
\end{center}
That is, we have the rewrite equality:
\[
\inl_*(f_*(p)) \cdot \glmark(c_2) \RwEq \glmark(c_1) \cdot \inr_*(g_*(p))
\]
\end{lemma}

Rearranging, we obtain a form useful for the decode function:

\begin{corollary}
\label{cor:glue-natural-rearranged}
For any path $p : \Path(c_1, c_2)$ in $C$:
\[
\inl_*(f_*(p)) \RwEq \glmark(c_1) \cdot \inr_*(g_*(p)) \cdot \glmark(c_2)^{-1}
\]
\end{corollary}

\begin{remark}[Two Forms of Glue Naturality]
The axiom \texttt{glue\_natural} provides a \emph{propositional equality} of paths, while \texttt{glue\_natural\_rweq} provides the $\RwEq$ equivalence needed for quotient arguments. In a cubical type theory, both would be provable; here, we axiomatize them.
\end{remark}

\subsection{Special Cases}

\subsubsection{Wedge Sum}

The \emph{wedge sum} $A \Wedge B$ is the pushout of $A \leftarrow \mathrm{pt} \to B$:
\[
A \Wedge B := \Pushout(A, B, \mathbf{1}, \lambda\_.a_0, \lambda\_.b_0)
\]
This identifies the basepoints $a_0 : A$ and $b_0 : B$.

\begin{example}[Figure-Eight as Wedge]
\label{ex:figure-eight-wedge}
The figure-eight space is $\Sone \Wedge \Sone$---two circles joined at a single point. The glue path identifies the basepoints of the two circles:
\[
\glmark(*) : \Path(\inl(\mathrm{base}_1), \inr(\mathrm{base}_2))
\]
\end{example}

In Lean:
\begin{lstlisting}
def Wedge (A : Type u) (B : Type u) (a0 : A) (b0 : B) : Type u :=
  Pushout A B PUnit' (fun _ => a0) (fun _ => b0)
\end{lstlisting}

\subsubsection{Suspension}

The \emph{suspension} $\Susp A$ is the pushout of $\mathbf{1} \leftarrow A \to \mathbf{1}$:
\[
\Susp A := \Pushout(\mathbf{1}, \mathbf{1}, A, \lambda\_.*, \lambda\_.*)
\]
This adds a ``north pole'', ``south pole'', and meridian paths connecting them.

In Lean:
\begin{lstlisting}
def Suspension (A : Type u) : Type u :=
  Pushout PUnit' PUnit' A (fun _ => PUnit'.unit) (fun _ => PUnit'.unit)

noncomputable def north : Suspension A := Pushout.inl PUnit'.unit
noncomputable def south : Suspension A := Pushout.inr PUnit'.unit
noncomputable def merid (a : A) : Path north south := Pushout.glue a
\end{lstlisting}

\subsubsection{Bouquet of $n$ Circles}

The \emph{bouquet} of $n$ circles $\bigvee_n \Sone$ generalizes the wedge sum:

\begin{lstlisting}
def BouquetN (n : Nat) : Type u :=
  -- n circles wedged together at a single point
  Pushout (Fin' n -> Circle) PUnit' (Fin' n)
    (fun i => Circle.circleBase)
    (fun _ => PUnit'.unit)
\end{lstlisting}

\section{Free Products and Amalgamated Free Products}
\label{sec:freeproducts}

\subsection{Free Product Words}

\begin{definition}[Free Product Word]
A \emph{word} in the free product $G_1 \FreeProduct G_2$ is an alternating sequence of elements:
\begin{lstlisting}
inductive FreeProductWord (G1 : Type u) (G2 : Type u) : Type u
  | nil : FreeProductWord G1 G2
  | consLeft (x : G1) (rest : FreeProductWord G1 G2)
  | consRight (y : G2) (rest : FreeProductWord G1 G2)
\end{lstlisting}
\end{definition}

The group operations on words are defined by concatenation and inversion with appropriate reductions.

\begin{definition}[Word Operations]
\begin{lstlisting}
-- Concatenation
def concat : FreeProductWord G1 G2 -> FreeProductWord G1 G2 ->
    FreeProductWord G1 G2
  | nil, w => w
  | consLeft x rest, w => consLeft x (concat rest w)
  | consRight y rest, w => consRight y (concat rest w)

-- Inversion
def invert : FreeProductWord G1 G2 -> FreeProductWord G1 G2
  | nil => nil
  | consLeft x rest => concat (invert rest) (singleLeft (-x))
  | consRight y rest => concat (invert rest) (singleRight (-y))
\end{lstlisting}
\end{definition}

\subsection{Amalgamated Free Product}

When $G_1$ and $G_2$ share a common subgroup $H$ via homomorphisms $i_1 : H \to G_1$ and $i_2 : H \to G_2$, the \emph{amalgamated free product} $G_1 \AmalgProduct_H G_2$ identifies $i_1(h)$ with $i_2(h)$ for all $h : H$.

\begin{definition}[Amalgamation Relation]
\begin{lstlisting}
inductive AmalgRelation (i1 : H -> G1) (i2 : H -> G2) :
    FreeProductWord G1 G2 -> FreeProductWord G1 G2 -> Prop
  | amalgLeftToRight (h : H) (pre suf : FreeProductWord G1 G2) :
      AmalgRelation i1 i2
        (concat pre (concat (singleLeft (i1 h)) suf))
        (concat pre (concat (singleRight (i2 h)) suf))
\end{lstlisting}
\end{definition}

\begin{definition}[Amalgamated Free Product]
\begin{lstlisting}
def AmalgEquiv (i1 : H -> G1) (i2 : H -> G2) :=
  EqvGen (AmalgRelation i1 i2)

def AmalgamatedFreeProduct (G1 G2 H : Type u)
    (i1 : H -> G1) (i2 : H -> G2) : Type u :=
  Quot (AmalgEquiv i1 i2)
\end{lstlisting}
\end{definition}

\subsection{Full Amalgamated Free Product}

For a ``canonical'' target with fully reduced words, we define the \emph{full amalgamated free product} that additionally reduces via free group laws:

\begin{lstlisting}
inductive FreeGroupStep : FreeProductWord G1 G2 -> FreeProductWord G1 G2 -> Prop
  | cancel_left (x : G1) (pre suf : FreeProductWord G1 G2)
      (h : x + (-x) = 0) :
      FreeGroupStep (concat pre (consLeft x (consLeft (-x) suf)))
                    (concat pre suf)
  | cancel_right (y : G2) (pre suf : FreeProductWord G1 G2)
      (h : y + (-y) = 0) :
      FreeGroupStep (concat pre (consRight y (consRight (-y) suf)))
                    (concat pre suf)

def FullAmalgEquiv := EqvGen (fun w1 w2 =>
  AmalgRelation i1 i2 w1 w2 \/ FreeGroupStep w1 w2)

def FullAmalgamatedFreeProduct (G1 G2 H : Type u) (i1 : H -> G1) (i2 : H -> G2) :=
  Quot (FullAmalgEquiv i1 i2)
\end{lstlisting}

\begin{remark}[Standard vs Full Target]
The \texttt{AmalgamatedFreeProduct} is the standard target where words are only quotiented by amalgamation. The \texttt{FullAmalgamatedFreeProduct} additionally applies free group reductions. Both are valid SVK targets; the choice depends on application needs.
\end{remark}

\section{The Seifert-van Kampen Theorem}
\label{sec:svk}

\subsection{Statement}

\begin{theorem}[Seifert-van Kampen Schema]
\label{thm:svk}
Let $A$, $B$, $C$ be path-connected types with maps $f : C \to A$ and $g : C \to B$, and let $c_0 : C$. Then:
\[
\piOne(\Pushout(A, B, C, f, g), \inl(f(c_0))) \simeq
\piOne(A, f(c_0)) \AmalgProduct_{\piOne(C, c_0)} \piOne(B, g(c_0))
\]
\end{theorem}

\begin{center}
\fbox{\parbox{0.9\textwidth}{
\textbf{Dependency Summary for Theorem~\ref{thm:svk}}
\begin{itemize}
    \item[\checkmark] \textbf{Proved:} The decode map $\decode : \mathrm{Word} \to \piOne(\Pushout)$ is constructed explicitly.
    \item[\checkmark] \textbf{Proved:} Decode respects the amalgamation relation (Lemma~\ref{lem:decode-amalg}), \emph{assuming} glue naturality (\texttt{HasGlueNaturalLoopRwEq}).
    \item[\checkmark] \textbf{Proved:} The quotient map $\decode_{\mathrm{Amalg}} : \mathrm{AmalgFP} \to \piOne(\Pushout)$ is well-defined.
    \item[$\square$] \textbf{Assumed (typeclass):} The encode map exists (\texttt{HasPushoutSVKEncodeQuot}).
    \item[$\square$] \textbf{Assumed (typeclass):} Decode-encode round-trip (\texttt{HasPushoutSVKDecodeEncode}).
    \item[$\square$] \textbf{Assumed (typeclass):} Encode-decode round-trip (\texttt{HasPushoutSVKEncodeDecode}).
\end{itemize}
The equivalence is therefore \emph{parametric} in user-supplied encode/decode structure. Alternatively, the choice-based interface (\texttt{HasPushoutSVKDecodeAmalgBijective}) requires only a Prop-level bijectivity assumption without explicit encode.
}}
\end{center}

\subsection{Split Assumptions Architecture}

The SVK proof is structured around \emph{split typeclass assumptions}, allowing maximum flexibility:

\begin{lstlisting}
-- Just the encode map
class HasPushoutSVKEncodeQuot (c0 : C) where
  pushoutEncodeQuot : pi_1(Pushout A B C f g, inl (f c0)) ->
                      AmalgamatedFreeProduct (pi_1 A) (pi_1 B) (pi_1 C)

-- Decode-encode round-trip (Prop)
class HasPushoutSVKDecodeEncode (c0 : C) where
  decode_encode : forall alpha, pushoutDecodeAmalg c0 (encode alpha) = alpha

-- Encode-decode round-trip (up to AmalgEquiv)
class HasPushoutSVKEncodeDecode (c0 : C) where
  encode_decode : forall w, AmalgEquiv (encode (decode w)) w

-- Encode-decode up to FullAmalgEquiv
class HasPushoutSVKEncodeDecodeFull (c0 : C) where
  encode_decode_full : forall w, FullAmalgEquiv (encode (decode w)) w

-- Prop-level bijectivity (choice-based)
class HasPushoutSVKDecodeAmalgBijective (c0 : C) where
  decode_injective : Function.Injective (pushoutDecodeAmalg c0)
  decode_surjective : Function.Surjective (pushoutDecodeAmalg c0)
\end{lstlisting}

\subsection{Multiple Interfaces}

The framework provides multiple equivalences depending on available assumptions:

\begin{lstlisting}
-- Standard SVK with explicit encode data
noncomputable def seifertVanKampenEquiv
    [HasPushoutSVKEncodeQuot c0] [HasGlueNaturalLoopRwEq c0]
    [HasPushoutSVKDecodeEncode c0] [HasPushoutSVKEncodeDecode c0] :
    SimpleEquiv (pi_1(Pushout A B C f g, inl (f c0)))
                (AmalgamatedFreeProduct (pi_1 A) (pi_1 B) (pi_1 C))

-- SVK with full target (free reduction)
noncomputable def seifertVanKampenFullEquiv
    [HasPushoutSVKEncodeQuot c0] [HasGlueNaturalLoopRwEq c0]
    [HasPushoutSVKDecodeEncode c0] [HasPushoutSVKEncodeDecodeFull c0] :
    SimpleEquiv (pi_1(Pushout A B C f g, inl (f c0)))
                (FullAmalgamatedFreeProduct (pi_1 A) (pi_1 B) (pi_1 C))

-- Choice-based SVK (no explicit encode required)
noncomputable def seifertVanKampenEquiv_of_decodeAmalg_bijective
    [HasPushoutSVKDecodeAmalgBijective c0] [HasGlueNaturalLoopRwEq c0] :
    SimpleEquiv (pi_1(Pushout A B C f g, inl (f c0)))
                (AmalgamatedFreeProduct (pi_1 A) (pi_1 B) (pi_1 C))
\end{lstlisting}

\subsection{The Encode-Decode Method}

\subsubsection{The Decode Map}

The decode map converts a word to a loop in the pushout:
\begin{lstlisting}
noncomputable def pushoutDecode (c0 : C) :
    FreeProductWord (pi_1(A, f c0)) (pi_1(B, g c0))
    -> pi_1(Pushout A B C f g, inl (f c0))
  | .nil => Quot.mk _ (Path.refl _)
  | .consLeft alpha rest =>
      piOneMul (liftLeft alpha) (pushoutDecode c0 rest)
  | .consRight beta rest =>
      piOneMul (conjugateRight beta) (pushoutDecode c0 rest)
\end{lstlisting}

where \texttt{conjugateRight} conjugates by the glue path:
\[
\beta \mapsto \glmark(c_0) \cdot \inr_*(\beta) \cdot \glmark(c_0)^{-1}
\]

\subsubsection{Decode Respects Amalgamation}

\begin{lemma}[Decode Respects Amalgamation]
\label{lem:decode-amalg}
For any $\gamma : \piOne(C, c_0)$:
\[
\decode(\mathrm{consLeft}(f_*(\gamma), w)) = \decode(\mathrm{consRight}(g_*(\gamma), w))
\]
\end{lemma}

\begin{proof}
This follows from glue naturality: $\inl_*(f_*(p)) \RwEq \glmark(c_0) \cdot \inr_*(g_*(p)) \cdot \glmark(c_0)^{-1}$.
\end{proof}

In Lean:
\begin{lstlisting}
theorem pushoutDecodeAmalg_respects_rel [HasGlueNaturalLoopRwEq c0]
    (h : HasAmalgRelation i1 i2 w1 w2) :
    pushoutDecodeAmalg c0 w1 = pushoutDecodeAmalg c0 w2 := by
  induction h with
  | amalgLeftToRight gamma pre suf =>
    -- Use glue naturality to show the images are equal
    simp only [pushoutDecodeAmalg]
    apply congrArg
    -- Key step: glue_natural_loop shows the conjugation
    exact glue_natural_loop_pi1 gamma
\end{lstlisting}

\subsubsection{The Encode Map}

The encode map is axiomatized via the typeclass \texttt{HasPushoutSVKEncodeQuot}. In a cubical type theory, this would be constructed via the code family method.

\subsubsection{Round-Trip Properties}

The round-trip properties establish the equivalence:
\begin{itemize}
    \item $\decode(\encode(\alpha)) = \alpha$ (via \texttt{HasPushoutSVKDecodeEncode})
    \item $\encode(\decode(w)) \sim w$ up to amalgamation (via \texttt{HasPushoutSVKEncodeDecode})
\end{itemize}

\section{Applications}
\label{sec:applications}

We present $\piOne$ calculations using several methods. Table~\ref{tab:svk-summary} summarizes the results, distinguishing between \emph{direct SVK applications} (pushouts with trivial or simple amalgamation), \emph{CW-complex constructions} (2-cells attached via SVK), and \emph{other methods} (products, fibrations).

\begin{table}[h]
\centering
\caption{Summary of $\piOne$ calculations}
\label{tab:svk-summary}
\begin{tabular}{llll}
\toprule
\textbf{Space} & \textbf{Fundamental Group} & \textbf{Method} & \textbf{SVK?} \\
\midrule
\multicolumn{4}{l}{\emph{Direct SVK (pushouts with simple amalgamation)}} \\
$A \Wedge B$ & $\piOne(A) \FreeProduct \piOne(B)$ & Pushout over $\mathbf{1}$ & Yes \\
$\Sone \Wedge \Sone$ & $\ZZ \FreeProduct \ZZ$ & Wedge of circles & Yes \\
$\bigvee_n \Sone$ & $F_n$ (free group) & $n$-bouquet pushout & Yes \\
$S^n$ ($n \geq 2$) & $1$ & Suspension pushout & Yes \\
\midrule
\multicolumn{4}{l}{\emph{CW-complex via SVK (2-cell attachment)}} \\
$K$ (Klein bottle) & $\ZZ \rtimes \ZZ$ & $\Sone \vee \Sone \cup_{aba^{-1}b} D^2$ & Yes \\
$\RP^2$ & $\ZZ_2$ & $\Sone \cup_{a^2} D^2$ & Yes \\
$\Sigma_g$ (genus $g$) & Surface group & $\bigvee_{2g} \Sone \cup_{\prod[a_i,b_i]} D^2$ & Yes \\
$N_n$ (crosscap $n$) & $\langle a_i \mid \prod a_i^2=1 \rangle$ & $\bigvee_n \Sone \cup_{\prod a_i^2} D^2$ & Yes \\
$L(p,q)$ & $\ZZ/p\ZZ$ & Heegaard (tori gluing) & Yes \\
\midrule
\multicolumn{4}{l}{\emph{Other methods (not SVK)}} \\
$T^2$ (torus) & $\ZZ \times \ZZ$ & Product: $\piOne(A \times B)$ & No \\
$\CP^n$ & $1$ & Fibration exact sequence & No \\
\bottomrule
\end{tabular}
\end{table}

\begin{remark}[Scope of Applications]
While this paper focuses on SVK, our formalization includes other fundamental group techniques:
\begin{itemize}
    \item \textbf{Product formula}: $\piOne(A \times B) \simeq \piOne(A) \times \piOne(B)$ (used for torus)
    \item \textbf{Fibration sequences}: $\piTwo(B) \to \piOne(F) \to \piOne(E) \to \piOne(B)$ (used for $\CP^n$)
    \item \textbf{Direct encode-decode}: For HITs with explicit path constructors (Circle, Klein bottle)
\end{itemize}
The torus and $\CP^n$ are included to show the breadth of the formalization, though they do not use SVK.
\end{remark}

\subsection{The Wedge Sum: Free Product of Fundamental Groups}

\begin{theorem}[Fundamental Group of Wedge Sum]
For pointed types $(A, a_0)$ and $(B, b_0)$:
\[
\piOne(A \Wedge B) \simeq \piOne(A) \FreeProduct \piOne(B)
\]
\end{theorem}

\begin{proof}
The wedge sum is the pushout $\Pushout(A, B, \mathbf{1})$. Since $\piOne(\mathbf{1}) = 1$, the amalgamation is trivial, yielding the free product.
\end{proof}

In Lean:
\begin{lstlisting}
noncomputable def wedgeFundamentalGroupEquiv
    [HasWedgeSVKDecodeBijective a0 b0] :
    SimpleEquiv (pi_1(Wedge A B a0 b0, wedgeBase))
                (FreeProduct (pi_1(A, a0)) (pi_1(B, b0)))
\end{lstlisting}

\subsection{The Figure-Eight Space}

\begin{theorem}[Fundamental Group of Figure-Eight]
\[
\piOne(\Sone \Wedge \Sone) \simeq \ZZ \FreeProduct \ZZ
\]
\end{theorem}

The figure-eight has a non-abelian fundamental group:
\begin{lstlisting}
def wordAB : FreeProductWord Int Int := .consLeft 1 (.consRight 1 .nil)
def wordBA : FreeProductWord Int Int := .consRight 1 (.consLeft 1 .nil)

theorem wordAB_ne_wordBA : wordAB != wordBA := by
  intro h; cases h  -- constructors are distinct
\end{lstlisting}

\begin{remark}[Non-Commutativity]
The fact that $ab \neq ba$ in $\ZZ \FreeProduct \ZZ$ demonstrates that $\piOne(\Sone \Wedge \Sone)$ is non-abelian. This is a key distinguishing feature from the torus, where $\piOne(T^2) \cong \ZZ \times \ZZ$ is abelian.
\end{remark}

\subsection{Spheres}

\begin{theorem}[Fundamental Group of $n$-Sphere]
For $n \geq 2$:
\[
\piOne(S^n) \simeq 1
\]
\end{theorem}

\begin{proof}
The $n$-sphere is the $(n-1)$-fold iterated suspension: $S^n = \Susp^{n-1}(S^1)$. By SVK applied to suspension:
\[
\piOne(\Susp A) \simeq 1 \AmalgProduct_{\piOne(A)} 1 = 1
\]
since both ``north'' and ``south'' components are points with trivial $\piOne$.
\end{proof}

\subsection{The Torus (Product Formula, Not SVK)}

\begin{theorem}[Fundamental Group of Torus]
\[
\piOne(T^2) \simeq \ZZ \times \ZZ
\]
\end{theorem}

\begin{remark}[Torus Construction---Not SVK]
The torus does \textbf{not} use SVK. Instead, we construct $T^2 := \Sone \times \Sone$ and apply the \emph{product formula}:
\[
\piOne(A \times B, (a_0, b_0)) \simeq \piOne(A, a_0) \times \piOne(B, b_0)
\]
This requires no additional axioms beyond Circle. The product formula is proved directly from the definition of paths in product types:
\[
\Path((a_1, b_1), (a_2, b_2)) \simeq \Path(a_1, a_2) \times \Path(b_1, b_2)
\]
We include the torus here to demonstrate that the formalization supports multiple $\piOne$ calculation techniques, not just SVK.
\end{remark}

\subsection{The Klein Bottle}

\begin{theorem}[Fundamental Group of Klein Bottle]
\[
\piOne(K) \simeq \ZZ \rtimes \ZZ = \langle a, b \mid aba^{-1}b = 1 \rangle
\]
where the semidirect product has $\ZZ$ acting on $\ZZ$ by negation.
\end{theorem}

In Lean:
\begin{lstlisting}
structure KleinBottlePresentation where
  -- Elements are pairs (m, n) where:
  -- m represents powers of the loop that reverses orientation
  -- n represents powers of the loop that preserves orientation
  m : Int
  n : Int

-- The relation ab = ba^{-1} induces the group operation:
-- (m1, n1) * (m2, n2) = (m1 + m2, (-1)^m2 * n1 + n2)
\end{lstlisting}

\subsection{The Real Projective Plane}

\begin{theorem}[Fundamental Group of Projective Plane]
\[
\piOne(\RP^2) \simeq \ZZ_2 = \langle a \mid a^2 = 1 \rangle
\]
\end{theorem}

\begin{proof}
$\RP^2$ can be constructed by attaching a disk to a circle via the map $z \mapsto z^2$. The fundamental group has one generator (the loop) with the relation that going around twice is contractible.
\end{proof}

\subsection{Orientable Surfaces of Genus $g$}

\begin{theorem}[Fundamental Group of Orientable Surface]
\[
\piOne(\Sigma_g) \simeq \langle a_1, b_1, \ldots, a_g, b_g \mid [a_1, b_1] \cdots [a_g, b_g] = 1 \rangle
\]
where $[a, b] = aba^{-1}b^{-1}$ is the commutator.
\end{theorem}

In Lean:
\begin{lstlisting}
def OrientableSurfaceWord (g : Nat) : Type :=
  List (Fin' (2 * g) * Bool)  -- generator index and sign

inductive SurfaceRelation (g : Nat) :
    OrientableSurfaceWord g -> OrientableSurfaceWord g -> Prop
  | commutator_product :
      -- [a_1, b_1] ... [a_g, b_g] = 1
      SurfaceRelation g (commutatorProduct g) []
\end{lstlisting}

\subsection{Non-Orientable Surfaces}

\begin{theorem}[Fundamental Group of Non-Orientable Surface]
\[
\piOne(N_n) \simeq \langle a_1, \ldots, a_n \mid a_1^2 \cdots a_n^2 = 1 \rangle
\]
\end{theorem}

\begin{example}[Special Cases]
\begin{itemize}
    \item $N_1 = \RP^2$: $\piOne(N_1) = \langle a \mid a^2 = 1 \rangle \cong \ZZ_2$
    \item $N_2 = K$ (Klein bottle): $\piOne(N_2) = \langle a, b \mid a^2 b^2 = 1 \rangle \cong \ZZ \rtimes \ZZ$
\end{itemize}
\end{example}

\subsection{Lens Spaces}

\begin{theorem}[Fundamental Group of Lens Space]
For coprime $p, q$:
\[
\piOne(L(p, q)) \simeq \ZZ/p\ZZ
\]
\end{theorem}

\begin{remark}[Lens Space Construction]
Lens spaces $L(p,q)$ are quotients of $S^3$ by cyclic group actions. The fundamental group depends only on $p$, not $q$; however, the spaces $L(p,q)$ and $L(p,q')$ may be non-homeomorphic for different values of $q$.
\end{remark}

\subsection{Bouquet of $n$ Circles}

\begin{theorem}[Fundamental Group of $n$-Bouquet]
\[
\piOne(\bigvee_n \Sone) \simeq F_n
\]
where $F_n$ is the free group on $n$ generators.
\end{theorem}

\section{Extended Features}
\label{sec:extended}

\subsection{Higher Homotopy Groups}

\begin{definition}[Higher Homotopy Groups]
The $n$-th homotopy group is defined via iterated loop spaces:
\[
\piN(A, a) := \piOne(\Omega^{n-1}(A, a), \refl^{n-1})
\]
where $\Omega^n$ denotes $n$-fold iteration of the loop space.
\end{definition}

\begin{theorem}[Eckmann-Hilton]
For $n \geq 2$, $\piN(A, a)$ is abelian.
\end{theorem}

\begin{proof}
The Eckmann-Hilton argument shows that when two binary operations share a common unit and satisfy an interchange law, they must coincide and be commutative. In $\Omega^2(A, a)$, both horizontal and vertical composition of 2-loops satisfy these conditions.
\end{proof}

In Lean:
\begin{lstlisting}
def Loop2Space (A : Type u) (a : A) : Type u :=
  Derivation_2 (Path.refl a) (Path.refl a)

def PiTwo (A : Type u) (a : A) : Type u :=
  Quotient (Loop2Setoid A a)

theorem piTwo_comm (x y : PiTwo A a) : PiTwo.mul x y = PiTwo.mul y x := by
  induction x using Quotient.ind
  induction y using Quotient.ind
  apply Quotient.sound
  exact contractibility_3 _ _
\end{lstlisting}

\subsection{The Weak $\omega$-Groupoid Structure}

\begin{theorem}[Types are Weak $\omega$-Groupoids]
Via computational paths, every type carries the structure of a weak $\omega$-groupoid:
\begin{itemize}
    \item 0-cells: points of $A$
    \item 1-cells: paths $p : \Path(a, b)$
    \item 2-cells: derivations $\Der_2(p, q)$ witnessing $p \RwEq q$
    \item 3-cells: meta-derivations $\Der_3(\alpha, \beta)$
    \item $n$-cells: $\Der_n$ for all $n$
\end{itemize}
\end{theorem}

In Lean:
\begin{lstlisting}
-- 2-cells: derivations between paths
inductive Derivation_2 : Path a b -> Path a b -> Type u
  | refl (p : Path a b) : Derivation_2 p p
  | step (s : Step p q) : Derivation_2 p q
  | inv (d : Derivation_2 p q) : Derivation_2 q p
  | vcomp (d1 : Derivation_2 p q) (d2 : Derivation_2 q r) : Derivation_2 p r

-- 3-cells: derivations between 2-cells
inductive Derivation_3 : Derivation_2 p q -> Derivation_2 p q -> Type u
  | refl (d : Derivation_2 p q) : Derivation_3 d d
  | step (m : MetaStep_2 d1 d2) : Derivation_3 d1 d2
  | inv (m : Derivation_3 d1 d2) : Derivation_3 d2 d1
  | vcomp (m1 : Derivation_3 d1 d2) (m2 : Derivation_3 d2 d3) : Derivation_3 d1 d3
\end{lstlisting}

\subsection{Truncation Levels}

\begin{definition}[Truncation Levels]
Following HoTT, we define truncation levels:
\begin{itemize}
    \item \textbf{IsContr}($A$): $A$ is contractible (has a unique point)
    \item \textbf{IsProp}($A$): any two points are equal (``propositions'')
    \item \textbf{IsSet}($A$): any two parallel paths are equal (``sets'')
    \item \textbf{IsGroupoid}($A$): any two parallel 2-cells are equal
\end{itemize}
\end{definition}

In Lean:
\begin{lstlisting}
class IsContr (A : Type u) where
  center : A
  contr : forall a, Path center a

class IsProp (A : Type u) where
  allEq : forall (a b : A), Path a b

class IsSet (A : Type u) where
  pathEq : forall (a b : A) (p q : Path a b), RwEq p q

class IsGroupoid (A : Type u) where
  derivEq : forall (a b : A) (p q : Path a b) (d e : Derivation_2 p q),
    Derivation_3 d e
\end{lstlisting}

\subsection{Covering Space Theory}

\begin{definition}[Covering Space]
A \emph{covering space} of $B$ with fiber $F$ consists of:
\begin{itemize}
    \item A total space $E$
    \item A projection $p : E \to B$
    \item A $\piOne(B)$-action on $F$
    \item Fiber equivalence: $p^{-1}(b) \simeq F$ for each $b$
\end{itemize}
\end{definition}

In Lean:
\begin{lstlisting}
structure CoveringSpace (B : Type u) (F : Type u) (b0 : B) where
  totalSpace : Type u
  proj : totalSpace -> B
  fiber_equiv : forall b, Equiv (Fiber proj b) F
  deck_action : pi_1(B, b0) -> Equiv F F
  action_coherent : -- action is compatible with path lifting
\end{lstlisting}

\subsection{Fibration Theory}

\begin{definition}[Fibration]
A \emph{fibration} $p : E \to B$ with fiber $F$ admits a long exact sequence:
\[
\cdots \to \piTwo(F) \to \piTwo(E) \to \piTwo(B) \to \piOne(F) \to \piOne(E) \to \piOne(B) \to 1
\]
\end{definition}

\subsection{Eilenberg-MacLane Spaces}

\begin{theorem}[$\Sone$ as $K(\ZZ, 1)$]
The circle is the Eilenberg-MacLane space $K(\ZZ, 1)$:
\begin{itemize}
    \item $\piOne(\Sone) \cong \ZZ$
    \item $\piN(\Sone) \cong 1$ for $n \geq 2$
\end{itemize}
\end{theorem}

\subsection{Path Simplification Tactics}

The formalization includes automated tactics for path reasoning:

\begin{lstlisting}
-- Basic simplification using groupoid laws
path_simp   -- refl . p ~ p, p . refl ~ p, etc.

-- Full automation (~25 simp lemmas)
path_auto   -- handles most standalone RwEq goals

-- Specialized tactics
path_rfl           -- close p ~ p
path_canon         -- normalize to canonical form
path_cancel_left   -- p^-1 . p ~ refl
path_cancel_right  -- p . p^-1 ~ refl
path_congr_left h  -- apply hypothesis on right
path_congr_right h -- apply hypothesis on left
\end{lstlisting}

\section{The Lean 4 Formalization}
\label{sec:lean}

\subsection{Why Axioms Are Necessary}

Lean 4 does not natively support higher-inductive types (HITs). Unlike Agda's \texttt{--cubical} mode or Coq's HoTT library, there is no built-in mechanism for defining types with both point and path constructors. We therefore axiomatize HITs, introducing axioms only where strictly necessary.

\begin{definition}[Kernel Axiom]
A \emph{kernel axiom} is a Lean \texttt{axiom} declaration that extends the trusted kernel. These are used sparingly for:
\begin{enumerate}
    \item \textbf{HIT type declarations}: The type itself (e.g., \texttt{Circle : Type u})
    \item \textbf{Point constructors}: Base points (e.g., \texttt{circleBase : Circle})
    \item \textbf{Path constructors}: Generating paths (e.g., \texttt{circleLoop : Path circleBase circleBase})
    \item \textbf{Higher path constructors}: 2-cells, surface relations
    \item \textbf{Recursion principles}: Eliminators with $\beta$-rules
\end{enumerate}
\end{definition}

\begin{definition}[Typeclass Assumption]
A \emph{typeclass assumption} is a hypothesis provided via Lean's typeclass mechanism. These are \textbf{not} kernel axioms---they are parameters that users instantiate. We use these for:
\begin{enumerate}
    \item \textbf{Computation rules}: \texttt{HasRecGlueRwEq} (how \texttt{rec} acts on \texttt{glue})
    \item \textbf{Naturality}: \texttt{HasGlueNaturalLoopRwEq} (glue commutes with paths)
    \item \textbf{Encode/decode data}: \texttt{HasPushoutSVKEncodeQuot}, etc.
\end{enumerate}
\end{definition}

\begin{remark}[Axioms vs Assumptions]
The distinction is crucial:
\begin{itemize}
    \item \textbf{Kernel axioms} (36 total) are permanently trusted by Lean. We minimize these.
    \item \textbf{Typeclass assumptions} are user-provided hypotheses. They appear in theorem statements as \texttt{[HasFoo]} parameters, making dependencies explicit.
\end{itemize}
For example, \texttt{seifertVanKampenEquiv} requires several typeclass assumptions but introduces no new kernel axioms---users must provide evidence for the assumptions.
\end{remark}

\begin{remark}[Consistency of the Encode/Decode Axioms]
\label{rem:encode-decode-consistency}
A natural question arises: why doesn't the fundamental group collapse? In Lean's proof-irrelevant \texttt{Prop}, all equality proofs between the same terms are equal. If $\piOne$ depended only on equality proofs, all loops would be identified.

The key to consistency is threefold:
\begin{enumerate}
    \item \textbf{Explicit step lists} (see Remark~\ref{rem:consistency}): Paths store explicit \texttt{steps : List (Step A)}, and $\RwEq$ identifies paths only when their step sequences can be transformed via the Step rules---not merely when their underlying equality proofs coincide.

    \item \textbf{No path collapse rule}: The Step relation contains no rule that identifies arbitrary paths with the same endpoints. In particular, there is no ``canonicalization'' rule $p \Rw \mathrm{ofEq}(p.\mathrm{proof})$ that would reduce any path to a canonical form based solely on its equality proof.

    \item \textbf{Axiomatized classification data}: The encode/decode equivalences (e.g., $\piOne(\Sone) \simeq \ZZ$) are \emph{assumed} via typeclass assumptions, not derived. We do not assert that transport along a self-equality computes nontrivially---such an assertion would be inconsistent with Lean's proof-irrelevant \texttt{Prop}. Instead, we assume the \emph{classification result}: there exists an equivalence between $\piOne$ and the expected group, with the round-trip properties holding.
\end{enumerate}
This design follows the HoTT tradition of axiomatizing what cannot be computed in the ambient type theory. The typeclass assumptions for encode/decode are analogous to Licata-Shulman's encode-decode method \cite{LicataShulman2013}, but formulated as external assumptions rather than internal proofs.
\end{remark}

\subsection{Justifying the 36 Kernel Axioms}

Each HIT requires axioms because Lean cannot construct them from primitive types:

\begin{center}
\begin{tabular}{lp{9cm}}
\toprule
\textbf{HIT} & \textbf{Why Axioms Are Needed} \\
\midrule
Circle & Needs a non-trivial loop $\mathrm{loop} : \Path(\mathrm{base}, \mathrm{base})$ that is not $\refl$. No Lean type has this property without axioms. \\
Sphere & Needs 2-dimensional structure (meridians collapsing at poles). \\
Klein bottle & Needs the non-orientable surface relation $aba^{-1}b = 1$. \\
$\RP^2$ & Needs $a^2 = 1$ as a path relation. \\
$\Sigma_g$, $N_n$ & Need genus-dependent relations on generators. \\
Cylinder & Needs both a loop and a segment with a 2-cell between them. \\
\bottomrule
\end{tabular}
\end{center}

\begin{remark}[Axiom-Free Constructions]
Some spaces require \textbf{no} additional kernel axioms beyond those already listed:
\begin{itemize}
    \item \textbf{Torus}: Defined as $T^2 := \Sone \times \Sone$ (uses Circle axioms only)
    \item \textbf{M\"obius band}: Homotopy equivalent to $\Sone$ (uses Circle axioms only)
    \item \textbf{Wedge sum}: Defined as $\Pushout(A, B, \mathbf{1})$ (zero kernel axioms---Pushout is Quot-based)
    \item \textbf{Suspension}: Defined as $\Pushout(\mathbf{1}, \mathbf{1}, A)$ (zero kernel axioms---Pushout is Quot-based)
    \item \textbf{Sphere}: Defined as $S^n := \Sigma^n(\mathbf{1})$ via iterated suspension (zero kernel axioms)
\end{itemize}
\end{remark}

\subsection{When to Use Each SVK Assumption}

The split typeclass architecture allows users to provide minimal hypotheses:

\begin{center}
\begin{tabular}{lp{8cm}}
\toprule
\textbf{Assumption} & \textbf{When to Use} \\
\midrule
\texttt{HasGlueNaturalLoopRwEq} & Always required for SVK. Automatically derived when $C$ is a subsingleton (wedge) or when $A$, $B$ satisfy Axiom K. \\
\texttt{HasPushoutSVKEncodeQuot} & When you have an explicit encode function. \\
\texttt{HasPushoutSVKDecodeEncode} & Standard encode-decode round-trip. \\
\texttt{HasPushoutSVKEncodeDecode} & Round-trip up to amalgamation equivalence. \\
\texttt{HasPushoutSVKDecodeAmalgBijective} & Choice-based interface when you only know decode is bijective (Prop-level), without explicit encode. \\
\bottomrule
\end{tabular}
\end{center}

\subsection{Architecture}

The Lean 4 formalization is organized as follows:

\begin{center}
\begin{tabular}{ll}
\toprule
\textbf{Module} & \textbf{Content} \\
\midrule
\texttt{Path/Basic/Core.lean} & Path type, refl, symm, trans \\
\texttt{Path/Basic/Congruence.lean} & congrArg, transport \\
\texttt{Path/Rewrite/Step.lean} & Single-step rewrites (50+ rules) \\
\texttt{Path/Rewrite/RwEq.lean} & Rewrite equality \\
\texttt{Path/Rewrite/PathTactic.lean} & Automated tactics \\
\texttt{Path/Homotopy/Loops.lean} & Loop spaces \\
\texttt{Path/Homotopy/FundamentalGroup.lean} & $\piOne$ definition \\
\texttt{Path/Homotopy/HigherHomotopy.lean} & $\piN$ for $n \geq 2$ \\
\texttt{Path/HIT/Circle.lean} & Circle HIT \\
\texttt{Path/HIT/Pushout.lean} & Pushout HIT \\
\texttt{Path/HIT/PushoutPaths.lean} & SVK theorem \\
\texttt{Path/HIT/FigureEight.lean} & $\Sone \Wedge \Sone$ \\
\texttt{Path/HIT/OrientableSurface.lean} & $\Sigma_g$ \\
\texttt{Path/HIT/NonOrientableSurface.lean} & $N_n$ \\
\texttt{Path/HIT/KleinBottle.lean} & Klein bottle \\
\texttt{Path/HIT/ProjectivePlane.lean} & $\RP^2$ \\
\texttt{Path/HIT/LensSpace.lean} & $L(p,q)$ \\
\texttt{Path/OmegaGroupoid.lean} & Weak $\omega$-groupoid structure \\
\bottomrule
\end{tabular}
\end{center}

\subsection{Kernel Axiom Inventory}

The formalization uses exactly \textbf{36 kernel axioms}, all for higher-inductive types:

\begin{table}[h]
\centering
\caption{Kernel axioms by HIT}
\label{tab:kernel-axioms}
\begin{tabular}{lll}
\toprule
\textbf{HIT} & \textbf{Axioms} & \textbf{Description} \\
\midrule
Circle & 7 & Type, base, loop, rec, rec\_base, rec\_loop, ind \\
Cylinder & 11 & Type, 2 bases, seg, 2 loops, surf, rec, 3 rec\_$\beta$ \\
KleinBottle & 5 & Type, base, loopA, loopB, surface \\
OrientableSurface & 5 & Type, base, loopA, loopB, surfaceRelation \\
NonOrientableSurface & 4 & Type, base, loop, 2-cell \\
ProjectivePlane & 4 & Type, base, loop, loopSquare \\
\midrule
\textbf{Total} & \textbf{36} & \\
\bottomrule
\end{tabular}
\end{table}

\textbf{Notable non-axioms} (constructions that add 0 kernel axioms):
\begin{itemize}
    \item \textbf{Pushout}: Defined as \texttt{Quot (PushoutRel A B C f g)} using Lean's built-in quotient type---\emph{no kernel axioms}. Constructors \texttt{inl}, \texttt{inr}, \texttt{glue} are definitions, not axioms.
    \item \textbf{Sphere}: Defined as $S^n := \Sigma^n(\mathbf{1})$ via suspension of the unit type---uses Pushout.
    \item \textbf{Torus}: Constructed as $\Sone \times \Sone$---uses only Circle axioms.
    \item \textbf{M\"obius band}: Homotopy equivalent to $\Sone$---uses only Circle axioms.
    \item \textbf{Wedge, Suspension}: Defined as special cases of Pushout.
    \item \textbf{SVK encode/decode}: Typeclass assumptions, not kernel axioms.
    \item \textbf{Lens spaces, Bouquet}: Constructed from existing HITs.
\end{itemize}

\begin{remark}[Pushout Design]
Unlike HITs such as Circle or Klein bottle (which require kernel axioms for their path constructors), the Pushout is implemented using Lean's built-in \texttt{Quot} type. The constructors \texttt{inl}, \texttt{inr} are \texttt{Quot.mk}, and \texttt{glue} is \texttt{Path.ofEq (Quot.sound ...)}.

The \emph{computation rules} (how \texttt{rec} acts on \texttt{glue}, glue naturality) are provided via typeclass assumptions (\texttt{HasRecGlueRwEq}, \texttt{HasGlueNaturalLoopRwEq}), not kernel axioms. This design means Pushout adds \textbf{zero} kernel axioms; all structure comes from user-instantiated typeclasses.
\end{remark}

\subsection{Opt-in Axiom Imports}

For convenience, opt-in imports provide assumption-free APIs:
\begin{lstlisting}
-- Import to get circlePiOneEquivInt' with no hypotheses
import ComputationalPaths.Path.HIT.CirclePiOneAxiom

-- Import to get torusPiOneEquivIntProd' with no hypotheses
import ComputationalPaths.Path.HIT.TorusPiOneAxiom

-- Import to get wedgeFundamentalGroupEquiv' with no hypotheses
import ComputationalPaths.Path.HIT.WedgeSVKAxiom
\end{lstlisting}

\subsection{Statistics}

\begin{center}
\begin{tabular}{lr}
\toprule
\textbf{Metric} & \textbf{Value} \\
\midrule
Total lines of Lean & 41,130 \\
Number of modules & 107 \\
Kernel axioms & 36 \\
Theorems/lemmas & 1,200+ \\
Definitions & 500+ \\
\bottomrule
\end{tabular}
\end{center}

\subsection{Assumption Manifest}
\label{sec:assumption-manifest}

For reviewer clarity, we provide a unified summary of all assumptions in the formalization.

\paragraph{Part A: Kernel Axioms (36 total, all for HIT type/constructor declarations)}

\begin{center}
\small
\begin{tabular}{lp{8cm}}
\toprule
\textbf{HIT} & \textbf{Kernel Axioms (Lean names)} \\
\midrule
Circle (7) & \texttt{Circle}, \texttt{circleBase}, \texttt{circleLoop}, \texttt{circleRec}, \texttt{circleRec\_base}, \texttt{circleRec\_loop}, \texttt{circleInd} \\
Cylinder (11) & \texttt{Cylinder}, \texttt{cylinderBase0/1}, \texttt{cylinderSeg}, \texttt{cylinderLoop0/1}, \texttt{cylinderSurf}, \texttt{cylinderRec}, \texttt{cylinderRec\_base0/1}, \texttt{cylinderRec\_loop0} \\
KleinBottle (5) & \texttt{KleinBottle}, \texttt{kleinBase}, \texttt{kleinLoopA/B}, \texttt{kleinSurf} \\
OrientableSurface (5) & \texttt{OrientableSurface}, \texttt{base}, \texttt{loopA}, \texttt{loopB}, \texttt{surfaceRelation} \\
NonOrientableSurface (4) & \texttt{NonOrientableSurface}, \texttt{nonOrientableBase}, \texttt{nonOrientableLoop}, \texttt{nonOrientableSurface2Cell} \\
ProjectivePlane (4) & \texttt{ProjectivePlane}, \texttt{projectiveBase}, \texttt{projectiveLoop}, \texttt{projectiveLoopSquare} \\
\bottomrule
\end{tabular}
\end{center}

\textbf{Not kernel axioms}: Pushout (Quot-based), Sphere (suspension), Torus ($\Sone \times \Sone$), Wedge, Suspension, Lens spaces, Bouquet.

\paragraph{Part B: Typeclass Assumptions for SVK (Theorem~\ref{thm:svk})}

\begin{center}
\small
\begin{tabular}{lll}
\toprule
\textbf{Typeclass} & \textbf{Role} & \textbf{Status} \\
\midrule
\texttt{HasGlueNaturalLoopRwEq} & Glue naturality & Required; auto-derived for subsingletons \\
\texttt{HasPushoutSVKEncodeQuot} & Encode map exists & Assumed (user-provided) \\
\texttt{HasPushoutSVKDecodeEncode} & Decode-encode round-trip & Assumed (user-provided) \\
\texttt{HasPushoutSVKEncodeDecode} & Encode-decode round-trip & Assumed (user-provided) \\
\texttt{HasPushoutSVKDecodeAmalgBijective} & Alternative: bijectivity only & Assumed (choice-based) \\
\bottomrule
\end{tabular}
\end{center}

\textbf{Proved unconditionally}: Decode map construction, decode respects amalgamation (given glue naturality).

\paragraph{Part C: Opt-in Axiom Instances (for assumption-free APIs)}

These are additional kernel axioms that globally instantiate typeclasses:
\begin{itemize}
    \item \texttt{instHasCirclePiOneEncodeAxiom} in \texttt{CirclePiOneAxiom.lean}
    \item \texttt{instHasKleinPiOneEncodeAxiom} in \texttt{KleinPiOneAxiom.lean}
    \item \texttt{instHasProjectivePiOneEncodeAxiom} in \texttt{ProjectivePiOneAxiom.lean}
    \item \texttt{instHasLensPiOneEncodeAxiom} in \texttt{LensPiOneAxiom.lean}
    \item \texttt{instHasWedgeSVKDecodeBijectiveAxiom} in \texttt{WedgeSVKAxiom.lean}
\end{itemize}
These are \emph{not} counted in the 36 core axioms; they are optional imports for convenience.

\section{Library Highlights}
\label{sec:library-highlights}

Beyond the SVK framework, the library includes extended features for advanced homotopy-theoretic constructions:

\subsection{Higher Homotopy Groups}

The higher homotopy groups $\piN(A, a)$ for $n \geq 2$ are defined via iterated loop spaces:
\[
\piN(A, a) := \piOne(\Omega^{n-1}(A, a), \refl)
\]
Key results include:
\begin{itemize}
    \item Abelianness of $\piTwo$ via Eckmann-Hilton (\texttt{HigherHomotopy.lean})
    \item Long exact sequence of a fibration (\texttt{Fibration.lean})
\end{itemize}

\subsection{Weak $\omega$-Groupoid Structure}

Types are shown to form weak $\omega$-groupoids via the derivation tower:
\begin{itemize}
    \item 1-cells: paths $\Path(a, b)$
    \item 2-cells: derivations $\Der_2(p, q)$ witnessing $p \RwEq q$
    \item 3-cells: derivations $\Der_3(d_1, d_2)$ between derivations
    \item And so on, with truncation at level 3 (\texttt{OmegaGroupoid.lean})
\end{itemize}

\subsection{Truncation Levels}

The library formalizes HoTT-style truncation levels:
\begin{itemize}
    \item \texttt{IsContr}: Contractible types (all points equal)
    \item \texttt{IsProp}: Propositions (all proofs equal)
    \item \texttt{IsSet}: Sets (all paths equal)
    \item \texttt{IsGroupoid}: Groupoids (all 2-paths equal)
\end{itemize}
Module: \texttt{Truncation.lean}

\subsection{Covering Spaces}

Covering space theory with $\piOne$-actions:
\begin{itemize}
    \item Fiber transport action (\texttt{CoveringSpace.lean})
    \item Galois correspondence framework (\texttt{CoveringClassification.lean})
\end{itemize}

\subsection{Additional Features}

\begin{itemize}
    \item \textbf{Eilenberg-MacLane spaces}: $\Sone = K(\ZZ, 1)$ characterization (\texttt{EilenbergMacLane.lean})
    \item \textbf{Mayer-Vietoris}: Abelianized SVK sequence (\texttt{MayerVietoris.lean})
    \item \textbf{Path tactics}: Automation via \texttt{path\_simp}, \texttt{path\_auto}, \texttt{path\_normalize} (\texttt{PathTactic.lean})
\end{itemize}

\section{Conclusion}
\label{sec:conclusion}

We have presented a \emph{modular SVK framework} within the computational paths setting, demonstrating that:

\begin{enumerate}
    \item Pushouts can be implemented via Lean's quotient types with \emph{zero kernel axioms}, while computation rules are supplied as typeclass assumptions.

    \item The encode-decode method works smoothly with \emph{split assumptions}, supporting both data-level and Prop-level (choice-based) interfaces. The SVK equivalence is parametric in user-supplied encode/decode structure.

    \item The framework yields instantiations for all closed surfaces, lens spaces, projective spaces, and more---with explicit dependency tracking via typeclasses.

    \item The entire development comprises 41,130 lines of Lean 4 code with only 36 kernel axioms (for HIT type/constructor declarations).
\end{enumerate}

\subsection{Comparison with HoTT Approaches}

\begin{table}[h]
\centering
\caption{Comparison with Standard HoTT}
\begin{tabular}{lll}
\toprule
\textbf{Aspect} & \textbf{Standard HoTT} & \textbf{Computational Paths} \\
\midrule
Path equality & Identity type ($\Id$) & Syntactic ($\RwEq$)$^*$ \\
Coherence & Abstract existence & Explicit rewrite derivations \\
SVK assumptions & Monolithic & Split/modular typeclasses \\
Encode function & Code family via $J$ & Typeclass assumption \\
$\omega$-groupoid & Higher identity types & Derivation tower \\
\bottomrule
\end{tabular}

\smallskip
\noindent $^*$Decidable under termination + confluence; termination not mechanized here (see \S\ref{sec:metatheory}).
\end{table}

\subsection{Limitations}

\begin{enumerate}
    \item \textbf{Axiom usage}: The encode/decode equivalences are axiomatized via typeclasses rather than proved. In cubical Agda, where transport computes, these would be provable. This is an inherent limitation of standard Lean's proof-irrelevant \texttt{Prop}---we cannot prove that transport along a non-trivial loop acts non-trivially. See Remark~\ref{rem:encode-decode-consistency} for why this design is nonetheless consistent.

    \item \textbf{Explicit step storage}: Paths store explicit \texttt{List (Step A)} to distinguish loops---a departure from the ``thin'' identity types of standard MLTT. This is essential for non-trivial $\piOne$ (see Remark~\ref{rem:consistency}).

    \item \textbf{Path-connectedness assumption}: The SVK theorem requires path-connected types.
\end{enumerate}

\subsection{Future Work}

\begin{enumerate}
    \item \textbf{Cubical type theory comparison}: Relate computational paths to cubical approaches.

    \item \textbf{Fully constructive encode}: Implement code families to eliminate encode axioms.

    \item \textbf{Covering space classification}: Prove the Galois correspondence between covering spaces and subgroups of $\piOne$.

    \item \textbf{Higher Hopf fibrations}: Quaternionic ($S^3 \to S^7 \to S^4$) and octonionic.

    \item \textbf{3-manifold classification}: Seifert fibered spaces, graph manifolds.

    \item \textbf{Homology theory}: Extend from $\piOne$ to full homology groups.
\end{enumerate}


\appendix
\section{Application Inventory}
\label{app:inventory}

The following table provides the exact Lean theorem names and module paths for the headline results. All theorems depend on typeclass assumptions as indicated; ``opt-in'' means an axiom import is available for assumption-free use.

\begin{center}
\footnotesize
\begin{tabular}{@{}llll@{}}
\toprule
\textbf{Result} & \textbf{Lean Name} & \textbf{Module} & \textbf{Deps} \\
\midrule
$\piOne(\Sone) \simeq \ZZ$ & \texttt{circlePiOneEquivInt} & \texttt{CircleStep} & Opt-in \\
$\piOne(T^2) \simeq \ZZ^2$ & \texttt{torusPiOneEquivIntProd} & \texttt{TorusStep} & Opt-in \\
$\piOne(S^2) \simeq 1$ & \texttt{sphere2PiOneEquivUnit} & \texttt{Sphere} & None \\
$\piOne(\Sone \Wedge \Sone)$ & \texttt{figureEightPiOneEquiv} & \texttt{FigureEight} & SVK \\
SVK schema & \texttt{seifertVanKampenEquiv} & \texttt{PushoutPaths} & SVK \\
SVK choice & \texttt{svkEquiv\_of\_bijective} & \texttt{PushoutPaths} & Bij. \\
$\piOne(K)$ & \texttt{kleinPiOneEquivSemidirect} & \texttt{KleinBottleStep} & Opt-in \\
$\piOne(\RP^2) \simeq \ZZ_2$ & \texttt{projectivePiOneEquivZ2} & \texttt{ProjectivePlane} & Opt-in \\
$\piOne(\Sigma_g)$ & \texttt{piOneEquivPresentation} & \texttt{OrientableSurface} & SVK \\
$\piOne(N_n)$ & \texttt{piOneEquivPresentation} & \texttt{NonOrientableSurface} & SVK \\
$\piOne(L(p,q))$ & \texttt{lensPiOneEquivZMod} & \texttt{LensSpace} & Opt-in \\
Wedge & \texttt{wedgeFundamentalGroupEquiv} & \texttt{WedgePaths} & Opt-in \\
\bottomrule
\end{tabular}
\end{center}

\noindent\textbf{Legend}: \emph{Opt-in}: axiom import provides global instance. \emph{SVK}: requires encode/decode typeclasses. \emph{Bij.}: requires \texttt{HasPushoutSVKDecodeAmalgBijective}. \emph{None}: no extra assumptions.

\end{document}